\documentclass[review]{elsarticle}
\graphicspath{ {./figures/} }
\usepackage{hyperref}
\usepackage{xurl} 
\usepackage{float}
\usepackage{verbatim} 
\usepackage{apalike}
\restylefloat{figure}
\floatstyle{plain} 
\restylefloat{table}

\usepackage{xcolor}
\usepackage{amssymb}
\usepackage{amsthm}
\usepackage{amsmath}
\usepackage{booktabs}
\usepackage[T1]{fontenc}
\usepackage{lmodern}
\usepackage{multirow}
\usepackage{placeins}
\usepackage{longtable}

\newtheorem{theorem}{Theorem}
\newtheorem{definition}{Definition}
\newtheorem{example}{Example}

\newtheorem{proposition}{Proposition}

\journal{}

\biboptions{authoryear}

\begin{document}
\begin{frontmatter}

\title{A Class of Rényi-based Concentration Metrics for Kolmogorov–Arnold Geometry in Multilayer Perceptrons}

\author[label1]{Juan Carlos Fernández}
\ead{juancarlosfdezmlg@gmail.com}

\author[label2,label3]{Asunción Jiménez-Cordero\corref{cor1}}
\ead{asuncionjc@uma.es}

\author[label4]{José Francisco Paris}
\ead{jfparis@uma.es}

\cortext[cor1]{Corresponding author.}
\address[label1]{University of Málaga, Málaga, Spain}
\address[label2]{Department of Mathematical Analysis, Statistics and Operations Research and Applied Mathematics, Area of Statistics and Operations Research, University of Málaga, Málaga, Spain}
\address[label3]{OASYS research group, University of Málaga, Málaga, Spain}
\address[label4]{TELMA Institute, University of Málaga, Málaga, Spain}

\begin{abstract}
The interpretability of artificial intelligence models has become a central concern, especially for black-box models such as multilayer perceptrons (MLPs). Kolmogorov--Arnold Geometry has recently emerged as a framework for analysing the internal representations of these networks, by examining how the hidden layer locally responds to changes in the input through the minors of its Jacobian. The degree to which a few of these minors dominate the rest, known as their concentration, has so far been measured with the Participation Ratio. Although this metric has useful properties, its maximum value depends on the size of the network, so a given value cannot be read as high or low concentration without reference to the architecture. To overcome this limitation, in this paper we propose the $\alpha$-RC metrics, a family of Rényi-based concentration measures that compare the normalized magnitudes of the minors with the uniform distribution through the Rényi divergence of order $\alpha$. All the metrics in the family are bounded between 0 and 1 regardless of the architecture, with 0 indicating that the magnitude is evenly spread over the minors and 1 that it is carried by a single one. Moreover, the metrics are non-decreasing in $\alpha$, and for $\alpha = 2$ they reduce to an explicit transformation of the Participation Ratio. Computational experiments with single-hidden-layer MLPs trained on a synthetic dataset show the benefits of our proposal.
\end{abstract}

\begin{keyword}
Kolmogorov--Arnold geometry \sep R\'enyi divergence \sep Participation Ratio \sep Jacobian minors \sep multilayer perceptrons \sep neural network interpretability
\end{keyword}

\end{frontmatter}

\section{Introduction}

Understanding how neural networks internally represent information remains one of the central challenges in modern machine learning. While deep neural networks have achieved remarkable performance across a wide range of applications, their internal representations are still only partially understood. Developing mathematical tools capable of characterizing these representations is therefore an important step towards explaining why neural networks generalize successfully and under which conditions they exhibit efficient learning \citep{heuillet2021explainability}.

A recent line of research has revisited the celebrated Kolmogorov--Arnold representation theorem \citep{kolmogorov1957representation}, originally proved by Kolmogorov and later refined by Arnold, which states that every continuous multivariate function on the multidimensional unit cube can be represented as a finite superposition of continuous univariate functions. Although this theorem is mathematically elegant, its practical impact on machine learning remained limited for decades due to the highly irregular nature of the functions involved in the representation \citep{girosi1989representation, kuurkova1991kolmogorov}.

Interest in the theorem has recently been revived by the introduction of Kolmogorov--Arnold Networks (KANs) \citep{liu2025kan}, a novel neural network architecture inspired by the representation theorem. Beyond their architectural design, these works have stimulated a renewed interest in understanding the geometric structure underlying the Kolmogorov--Arnold representation itself.

In this context, Freedman introduced the notion of \emph{Kolmogorov--Arnold Geometry} (KAG), a geometric framework derived from the inner map appearing in the proof of the Kolmogorov--Arnold theorem \citep{freedman2024proof}. The central idea is that the local geometric properties of this map provide valuable information about how different variables interact within the representation. More importantly, the author hypothesized that similar geometric patterns may emerge spontaneously in neural network architectures that are not explicitly based on the Kolmogorov--Arnold construction, particularly in single-hidden-layer multilayer perceptrons (MLPs).


 To quantify these geometric properties, \citep{freedman2025spontaneous} proposed
several metrics based on the Jacobian of the inner map and its exterior powers,
whose entries are the $k \times k$ minors of the Jacobian. We say that these minors
are \emph{concentrated} when a few of them carry most of the total magnitude while
the remaining ones contribute little. For instance, the magnitudes of the matrix of minors $\left(\begin{smallmatrix} 9 & 0 \\ 0 & 1 \end{smallmatrix}\right)$ are
more concentrated than those of $\left(\begin{smallmatrix} 2.5 & 2.5 \\ 2.5 & 2.5 \end{smallmatrix}\right)$, although both add up to the same total.
Among the proposed metrics, the Participation Ratio (PR) has received particular
attention as a measure of this concentration, with lower values of the PR
indicating a higher concentration of the minors. Experimental evidence suggests that, under certain training regimes,
an increase in concentration with respect to initialization is associated with better
model performance. Interestingly, this relationship appears only in specific settings,
often referred to as the Goldilocks regime, while it disappears in others. Although the
reasons for this phenomenon are not yet fully understood, these observations suggest
that the PR captures relevant aspects of the Kolmogorov--Arnold Geometry emerging in
multilayer perceptrons.



Despite these promising results, the interpretation of the Participation Ratio
depends on the architecture under study. While its minimum value, attained when a
single minor carries all the magnitude, is the same for every network, its maximum
value, attained when the magnitude is uniformly spread over all minors, grows with
the number of hidden neurons. Consequently, the same PR value may lie close to the
maximum for a narrow network and far from it for a wider one, so its position within
the attainable range, and therefore its interpretation in terms of concentration,
changes with the architecture. \citep{freedman2025spontaneous} report the ratio
between the PR of the trained model and its value at initialization. This ratio
measures how concentration changes during training for a given network, but it does
not place the PR on a scale whose extremes have a fixed meaning.

The main objective of this paper is to address this limitation. We introduce a
family of normalized concentration metrics based on the R\'enyi entropy of order
$\alpha$, applied to the probability distribution obtained by normalizing the
magnitudes of the minors. The minimum and maximum values of these metrics are $0$
and $1$, regardless of the architecture, and so is their interpretation: a value of
$0$ corresponds, for every order $\alpha > 0$, to minimum concentration, with the magnitude uniformly spread over
all minors, and a value of $1$ to maximum concentration, with all the magnitude
carried by a single minor. The proposed metric is monotonically non-decreasing with
respect to $\alpha$, and the order $\alpha$ determines the weight given to the
different parts of the underlying distribution: small values of $\alpha$ take into account the proportion of non-zero entries, whereas large values emphasize its mode. In addition, for
$\alpha = 2$ the metric can be written as a function of the Participation Ratio,
establishing a direct theoretical connection between both quantities.

Finally, we illustrate the behaviour of the proposed metric through a series of
computational experiments on single-hidden-layer MLPs trained on a continuous
version of the XOR problem. We show that, unlike the PR, the proposed metric
places the concentration of the minors on a scale whose extremes have the same
meaning for every architecture, and we analyse how the choice of $\alpha$ affects
the resulting values.

\section{The Kolmogorov--Arnold Geometry}

This section reviews the main concepts underlying the notion of \emph{Kolmogorov--Arnold Geometry} (KAG), recently introduced by \citep{freedman2024proof} as a geometric framework inspired by the proof of the Kolmogorov--Arnold representation theorem. We first describe the geometric intuition behind this concept in Section~\ref{subs: from theorem to geometry} and then present the mathematical objects used to quantify it, namely the Jacobian and its exterior powers, in Section~\ref{subs: jacobians}. Section~\ref{subs: MLP and KAG} particularizes these objects to single-hidden-layer MLPs, the architecture studied in this work, by decomposing the network into an inner and an outer map and computing the minors of the Jacobian of its inner map. Finally, Section~\ref{subs: participation ratio} describes a particular metric proposed by \citep{freedman2025spontaneous}, called the Participation Ratio, which measures the concentration of these minors.

\subsection{From the Kolmogorov--Arnold theorem to Kolmogorov--Arnold Geometry} \label{subs: from theorem to geometry}


The Kolmogorov--Arnold representation theorem is one of the most celebrated results in approximation theory. Roughly speaking, it states that every continuous multivariate function on the multidimensional unit cube can be represented as a finite superposition of continuous univariate functions. While several variants and refinements of this theorem have been proposed in the literature, in this work we consider the formulation by \citep{morris2021hilbert}, which can be stated as follows:

\begin{theorem}[Kolmogorov--Arnold Theorem]
\label{thm: KA theorem}
Let $n \in \mathbb{N}$ with $n \geq 2$. There exist real numbers
$\lambda_1, \lambda_2, \ldots, \lambda_n$ and Lipschitz (and thus continuous)
functions $\varphi_j : [0,1] \to \mathbb{R}$,  $j = 1, \ldots, m$, with $m \geq 2n+1$, such that
for every continuous function $f : [0,1]^n \to \mathbb{R}$, there exists a continuous
function $g : \mathbb{R} \to \mathbb{R}$ satisfying, for each
$(x_1, x_2, \ldots, x_n) \in [0,1]^n$,
\begin{equation}\label{eq:KA}
f(x_1, x_2, \ldots, x_n)
=
\sum_{j=1}^{m}
g\!\left(
\lambda_1 \varphi_j(x_1)
+ \cdots +
\lambda_n \varphi_j(x_n)
\right).
\end{equation}
Furthermore, the constants $\lambda_i$ and the functions $\varphi_j$ can be chosen
independently of the function $f$.
\end{theorem}

We omit the formal proof of Theorem~\ref{thm: KA theorem} here. However, the interested reader is referred to \citep{morris2021hilbert} for a detailed proof and discussion.

Theorem~\ref{thm: KA theorem} can be read as a decomposition of $f$ into two maps.
The inner map $\Phi^{\mathrm{KA}} : [0,1]^n \to \mathbb{R}^m$ is defined
componentwise by
\begin{equation} \label{eq: inner map KA}
\Phi^{\mathrm{KA}}_j(x) = \sum_{i=1}^{n} \lambda_i\, \varphi_j(x_i),
\qquad j = 1, \dots, m,
\end{equation}
and the outer map $g^{\mathrm{KA}} : \mathbb{R}^m \to \mathbb{R}$ by
\begin{equation} \label{eq: outer map KA}
g^{\mathrm{KA}}(h) = \sum_{j=1}^{m} g(h_j),
\qquad h = (h_1, \dots, h_m)^T \in \mathbb{R}^m.
\end{equation}
With this notation, we define the KA representation
$f^{\mathrm{KA}} := g^{\mathrm{KA}} \circ \Phi^{\mathrm{KA}}$, and
Theorem~\ref{thm: KA theorem} states that $f^{\mathrm{KA}} = f$ on $[0,1]^n$. Since
the constants $\lambda_i$ and the functions $\varphi_j$ do not depend on $f$, the
inner map $\Phi^{\mathrm{KA}}$ is universal: the same map serves for every
continuous function, and only the outer map, through $g$, depends on the target
function.

Regarding its regularity, each component $\Phi^{\mathrm{KA}}_j$ is a finite linear
combination of Lipschitz functions, and hence Lipschitz. Therefore
$\Phi^{\mathrm{KA}}$ is a Lipschitz map, and by Rademacher's theorem it is
differentiable almost everywhere in $[0,1]^n$ with respect to the Lebesgue measure.
Its Jacobian matrix, and hence all the quantities derived from it in
Section~\ref{subs: jacobians}, are therefore defined almost everywhere, but not
necessarily at every point. Throughout this work, the Jacobian of
$\Phi^{\mathrm{KA}}$ is considered only at points where $\Phi^{\mathrm{KA}}$ is
differentiable.

Beyond its theoretical guarantee of universal approximation,
Theorem~\ref{thm: KA theorem} admits a compelling geometric interpretation through
the structure of the inner map $\Phi^{\mathrm{KA}}$. The recent work of
\citep{freedman2025spontaneous} characterizes the local geometric properties of
this map, a framework referred to as \emph{Kolmogorov--Arnold geometry (KAG)}. It does so
through differential quantities, mainly the Jacobian matrix and its exterior powers
(Section~\ref{subs: jacobians}), which provide information about how the hidden
representation varies with respect to the input variables.

\subsection{Jacobians, exterior powers and minor concentration} \label{subs: jacobians}

Consider the inner map $\Phi^{\mathrm{KA}} : [0,1]^n \to \mathbb{R}^m$ defined in
Equation~\eqref{eq: inner map KA}, where $n$ denotes the input dimension and $m$
the number of hidden coordinates. Given a point $\tilde{x} \in [0,1]^n$ at which
$\Phi^{\mathrm{KA}}$ is differentiable (see Section~\ref{subs: from theorem to geometry}), the
Jacobian matrix $J_{\Phi^{\mathrm{KA}}}(\tilde{x}) \in \mathcal{M}_{m \times n}(\mathbb{R})$
describes the local sensitivity of each hidden coordinate with respect to the input
variables at $\tilde{x}$. In other words, each row of the Jacobian characterizes
how the corresponding hidden coordinate locally reacts to infinitesimal
perturbations of the input. Specifically, the $(j, i)$-th entry of the Jacobian
matrix is the partial derivative of the $j$-th component function with respect to
the $i$-th input variable evaluated at $\tilde{x}$, namely
\begin{equation}
[J_{\Phi^{\mathrm{KA}}}(\tilde{x})]_{ji} \doteq
\frac{\partial \Phi^{\mathrm{KA}}_j(\tilde{x})}{\partial x_i},
\qquad j = 1, \dots, m, \quad i = 1, \dots, n.
\end{equation}

An important consequence of the proof of the Kolmogorov--Arnold theorem (Theorem~\ref{thm: KA theorem}) is that, in the ideal construction, many rows of the Jacobian vanish. Geometrically, this means that several hidden neurons become locally insensitive to perturbations of the input, producing a highly structured differential representation.

However, analysing individual neurons is not sufficient to characterize the complete geometry. It is also necessary to study how groups of hidden neurons jointly respond to variations of the input. To this end, \citep{freedman2025spontaneous} considered the exterior powers of the Jacobian matrix. The entries of the $k$-th exterior power correspond to the minors of order $k$, which quantify the joint variation of subsets of hidden neurons with respect to subsets of input variables.

To formalize this notion, we introduce the definition of the $k$-th exterior power of a matrix, which collects all possible $k \times k$ minors into a single operator.

\begin{definition} \label{def: exterior_power}
Let $A \in \mathcal{M}_{m \times n}(\mathbb{R})$ and let $1 \leq k \leq \min(m, n)$. Define the ordered index sets of cardinality $k$ as
\begin{equation}
    \mathcal{J}_{k, m} = \left\{ J = (j_1, j_2, \ldots, j_k) \in \{1, \ldots, m\}^k : 1 \leq j_1 < j_2 < \cdots < j_k \leq m \right\},
\end{equation}
and similarly $\mathcal{I}_{k, n}$ for subsets of $\{1, \ldots, n\}$ of size $k$. For any pair of index sets $J \in \mathcal{J}_{k, m}$ and $I \in \mathcal{I}_{k, n}$, let $A_{J, I} \in \mathcal{M}_{k \times k}(\mathbb{R})$ denote the submatrix formed by selecting the rows indexed by $J$ and the columns indexed by $I$.

The $k$-th exterior power of $A$, denoted by $P_A^k \in \mathcal{M}_{\binom{m}{k} \times \binom{n}{k}}(\mathbb{R})$, is the matrix whose entry corresponding to the row index $J \in \mathcal{J}_{k, m}$ and column index $I \in \mathcal{I}_{k, n}$ is given by the determinant of $A_{J, I}$:
\begin{equation}
    \left[ P_A^k \right]_{J, I} = \det(A_{J, I}) = \det \begin{pmatrix}
    A_{j_1, i_1} & A_{j_1, i_2} & \cdots & A_{j_1, i_k} \\
    A_{j_2, i_1} & A_{j_2, i_2} & \cdots & A_{j_2, i_k} \\
    \vdots & \vdots & \ddots & \vdots \\
    A_{j_k, i_1} & A_{j_k, i_2} & \cdots & A_{j_k, i_k}
    \end{pmatrix}
\end{equation}
\end{definition}




Applying Definition~\ref{def: exterior_power} to the Jacobian matrix $J_\Phi(x)$
yields its $k$-th exterior power $P^k_{J_\Phi(x)}$, whose entries admit a natural
geometric interpretation in terms of volume changes. For a linear map
$T : \mathbb{R}^n \to \mathbb{R}^n$ given by $T(x) = Ax$, with
$A \in \mathcal{M}_{n \times n}(\mathbb{R})$, an infinitesimal volume element $dV$
is mapped to an element of volume $|\det(A)|\,dV$. For a nonlinear map, the same
holds locally around a point $\tilde{x}$, with $A$ replaced by the Jacobian matrix at
$\tilde{x}$. This is the Jacobian determinant that appears in the change of variables
formula for multiple integrals.

Minors extend this idea to rectangular Jacobians and to subsets of variables. The
minor $\left[P^k_{J_\Phi(x)}\right]_{J, I}$ is the factor by which the oriented
$k$-dimensional volume of an infinitesimal parallelepiped spanned by the input
directions $I = (i_1, \ldots, i_k)$ changes when it is mapped by $J_\Phi(x)$ and
projected onto the hidden coordinates $J = (j_1, \ldots, j_k)$. Its absolute value
measures the magnitude of the volume change, while its sign indicates orientation.
For $k = 1$, the minors reduce to the entries of the Jacobian, and if $m = n$ and
$k = n$, the only minor is the full Jacobian determinant, a case that does not
arise in this work, where $m > n$. For instance, for a map with two inputs and
three hidden coordinates, the three minors of order two are the oriented areas of
the projections of the image of an infinitesimal square onto the planes spanned by
the hidden coordinates $(1, 2)$, $(1, 3)$ and $(2, 3)$.

The distribution of these minors plays a fundamental role in the characterization of the Kolmogorov--Arnold Geometry. When only a relatively small number of minors concentrates most of the total magnitude, the geometry is said to exhibit a high degree of concentration. Conversely, if the magnitude is distributed more uniformly among the minors, the corresponding geometric structure is considered less pronounced.

\subsection{MLP and its Kolmogorov--Arnold Geometry} \label{subs: MLP and KAG}

A fundamental question raised by \citep{freedman2025spontaneous} is whether
Kolmogorov--Arnold geometry emerges spontaneously during the standard training of
conventional neural network architectures. Unlike the construction underlying
Theorem~\ref{thm: KA theorem}, where the inner functions are explicitly designed,
standard multilayer perceptrons (MLPs) do not impose any KA structure
\textit{a priori}. In this context, the emergence of KA geometry does not refer to
the appearance of exact zeros. It refers to the development, during training, of
patterns similar to those of the KA construction: Jacobian rows that are
approximately zero, and a distribution of minors in which a few of them carry most
of the total magnitude. These patterns are assessed by comparing the trained model
with its initialization.

To evaluate whether a trained MLP exhibits a geometric structure compatible with
the KA framework, the network is first decomposed into an inner map and an outer
map. In the KA representation (Theorem~\ref{thm: KA theorem}), the inner map embeds
the input domain into a higher-dimensional space, where a relatively simple outer
map reconstructs the target function.

We consider a single-hidden-layer MLP (shallow MLP) with $n$ inputs, $m$ hidden
neurons and a single linear output neuron. Let
$x = (x_1, \dots, x_n)^T \in \mathbb{R}^n$ denote the input vector. The network is
parameterized by:
\begin{itemize}
    \item an input weight matrix $A \in \mathcal{M}_{m \times n}(\mathbb{R})$ and a
    bias vector $a \in \mathbb{R}^m$;
    \item an output weight vector $B \in \mathcal{M}_{1 \times m}(\mathbb{R})$ and a
    scalar bias $b \in \mathbb{R}$;
    \item a smooth activation function $\sigma : \mathbb{R} \to \mathbb{R}$. In this
    work we use the GeLU activation, $\sigma(t) = t\,F(t)$, where $F$ denotes the
    cumulative distribution function of the standard normal distribution. We
    denote by $\boldsymbol{\sigma} : \mathbb{R}^m \to \mathbb{R}^m$ its
    componentwise extension,
    $\boldsymbol{\sigma}(z) = (\sigma(z_1), \dots, \sigma(z_m))^T$ for
    $z = (z_1, \dots, z_m)^T \in \mathbb{R}^m$.
\end{itemize}
The input--output mapping $f^{\mathrm{MLP}} : \mathbb{R}^n \to \mathbb{R}$ is then
given by
\begin{equation}
f^{\mathrm{MLP}}(x) = B\,\boldsymbol{\sigma}(Ax + a) + b.
\end{equation}

Following the KA framework, we decompose $f^{\mathrm{MLP}}$ into two stages:
\begin{enumerate}
    \item \textbf{Inner map} $\Phi^{\mathrm{MLP}}$:
    \begin{equation} \label{eq: inner map MLP}
    \Phi^{\mathrm{MLP}} : \mathbb{R}^n \longrightarrow \mathbb{R}^m, \quad
    \Phi^{\mathrm{MLP}}(x) = \boldsymbol{\sigma}(z(x)),
    \end{equation}
    where $z(x) = Ax + a \in \mathbb{R}^m$ is the pre-activation vector.
    \item \textbf{Outer map} $g^{\mathrm{MLP}}$: the affine map acting on the hidden
    representation $h = \Phi^{\mathrm{MLP}}(x)$,
    \begin{equation} \label{eq: outer map MLP}
    g^{\mathrm{MLP}} : \mathbb{R}^m \longrightarrow \mathbb{R}, \quad
    g^{\mathrm{MLP}}(h) = B\,h + b.
    \end{equation}
\end{enumerate}
With this notation,
$f^{\mathrm{MLP}} = g^{\mathrm{MLP}} \circ \Phi^{\mathrm{MLP}}$, in direct analogy
with the decomposition $f^{\mathrm{KA}} = g^{\mathrm{KA}} \circ \Phi^{\mathrm{KA}}$
of Section~\ref{subs: from theorem to geometry}.

Since GeLU is smooth, $\Phi^{\mathrm{MLP}}$ is differentiable at every point. Its
Jacobian is therefore defined everywhere, unlike that of $\Phi^{\mathrm{KA}}$,
which, as discussed in Section~\ref{subs: from theorem to geometry}, is only guaranteed to exist
almost everywhere. For non-differentiable activations such as ReLU, the Jacobian
may not exist at points where some pre-activation crosses zero, and a convention
would be required to compute the metrics. Using the tools of
Section~\ref{subs: jacobians}, the entries of the Jacobian matrix
$J_{\Phi^{\mathrm{MLP}}}(x) \in \mathcal{M}_{m \times n}(\mathbb{R})$ at a point $x$ are:
\begin{equation}
[J_{\Phi^{\mathrm{MLP}}}(x)]_{ji} \doteq
\frac{\partial \Phi^{\mathrm{MLP}}_j(x)}{\partial x_i}
= \sigma'(z_j(x))\,A_{ji}, \qquad j = 1, \dots, m, \quad i = 1, \dots, n,
\end{equation}
where $z_j(x) = \sum_{i=1}^{n} A_{ji}\,x_i + a_j$ is the pre-activation of the
$j$-th hidden neuron and $\sigma'$ is the derivative of the scalar activation
function. For GeLU,
\begin{equation}
\sigma'(t) = F(t) + t \, \frac{1}{\sqrt{2\pi}}\, e^{-t^2/2}.
\end{equation}
Equivalently, $J_{\Phi^{\mathrm{MLP}}}(x) = D(x)\,A$, where
$D(x) = \operatorname{diag}\big(\sigma'(z_1(x)), \dots, \sigma'(z_m(x))\big)$.

In an MLP, row $j$ of $J_{\Phi^{\mathrm{MLP}}}(x)$ is approximately zero for one of
two reasons. Either the weight vector $A_{j\cdot}$ is close to zero (for instance,
due to weight sparsity or dead units), or $\sigma'(z_j(x))$ is close to zero at
$x$, as happens with GeLU for sufficiently negative pre-activations. The first
mechanism affects the row at every input, whereas the second depends on $x$. It is
this input-dependent mechanism that resembles the local insensitivity of hidden
coordinates in the KA construction.

Finally, for $1 \le k \le \min(m, n)$, the entries of the $k$-th exterior power of
the MLP Jacobian, $P^k_{J_{\Phi^{\mathrm{MLP}}}(x)} \in
\mathcal{M}_{\binom{m}{k} \times \binom{n}{k}}(\mathbb{R})$
(Definition~\ref{def: exterior_power}), are

\begin{equation} \label{eq: minors MLP}
\begin{split}
\left[ P^k_{J_{\Phi^{\mathrm{MLP}}}(x)} \right]_{J, I}
&= \det \begin{pmatrix}
\sigma'(z_{j_1}(x))\, A_{j_1, i_1} & \cdots & \sigma'(z_{j_1}(x))\, A_{j_1, i_k} \\
\vdots & \ddots & \vdots \\
\sigma'(z_{j_k}(x))\, A_{j_k, i_1} & \cdots & \sigma'(z_{j_k}(x))\, A_{j_k, i_k}
\end{pmatrix} \\
&= \left( \prod_{l=1}^{k} \sigma'(z_{j_l}(x)) \right) \det(A_{J, I}),
\end{split}
\end{equation}

\noindent where $J = (j_1, \dots, j_k) \in \mathcal{J}_{k, m}$ and
$I = (i_1, \dots, i_k) \in \mathcal{I}_{k, n}$. The second equality follows from
the multilinearity of the determinant with respect to its rows, and shows that a
hidden neuron with $\sigma'(z_j(x)) \approx 0$ makes every minor that involves row
$j$ approximately zero. These exterior powers quantify how the inner map of the MLP
locally transforms $k$-dimensional volume elements, providing the basis for
measuring minor concentration in neural networks. For simplicity, throughout the
rest of this work we denote $P^k_{J_{\Phi^{\mathrm{MLP}}}(x)}$ by $J^{(k)}(x)$.

\subsection{Participation Ratio} \label{subs: participation ratio}

Building on the exterior powers introduced in Section~\ref{subs: MLP and KAG}, we now
quantify the concentration of the minors of the inner map $\Phi^{\mathrm{MLP}}$ defined in
Equation~\eqref{eq: inner map MLP}. To this end, \citep{freedman2025spontaneous}
introduced the Participation Ratio (PR), one of the main quantitative tools for
detecting the spontaneous emergence of KA geometry in multilayer perceptrons. The PR
measures whether the total magnitude of the minors in $J^{(k)}(x)$ is concentrated
in a few of them or spread across many.

\begin{definition}[Participation Ratio] \label{def: PR}
Let $x \in \mathbb{R}^n$ be an input point such that $J^{(k)}(x) \neq 0$. The
order-$k$ Participation Ratio at $x$ is defined as
\begin{equation} \label{eq: PR}
\mathrm{PR}^{(k)}(x) := \frac{\|J^{(k)}(x)\|_1}{\|J^{(k)}(x)\|_2},
\end{equation}
where
\begin{equation*}
\|J^{(k)}(x)\|_1 = \sum_{J \in \mathcal{J}_{k,m}} \sum_{I \in \mathcal{I}_{k,n}}
\left| J^{(k)}_{J,I}(x) \right|
\end{equation*}
is the entrywise $\ell_1$ norm and
\begin{equation*}
\|J^{(k)}(x)\|_2 = \Bigg( \sum_{J \in \mathcal{J}_{k,m}} \sum_{I \in \mathcal{I}_{k,n}}
\left| J^{(k)}_{J,I}(x) \right|^2 \Bigg)^{1/2}
\end{equation*}
is the entrywise $\ell_2$ (Frobenius) norm. For an evaluation dataset
$\mathcal{D} = \{x^{(1)}, \dots, x^{(S)}\}$, the global Participation Ratio is
defined as the mean over all points of the dataset,
\begin{equation} \label{eq: global PR}
\mathrm{PR}^{(k)} := \frac{1}{S} \sum_{s=1}^{S} \mathrm{PR}^{(k)}(x^{(s)}).
\end{equation}
\end{definition}


Let $N_k = \binom{m}{k}\binom{n}{k}$ denote the total number of minors of order $k$.
Viewing $J^{(k)}(x)$ as a vector $v$ with $N_k$ entries, the inequalities
$\|v\|_2 \leq \|v\|_1 \leq \sqrt{N_k}\, \|v\|_2$ hold, the second one being a
consequence of the Cauchy--Schwarz inequality. Hence
\begin{equation} \label{eq: PR bounds}
1 \leq \mathrm{PR}^{(k)}(x) \leq \sqrt{N_k}.
\end{equation}
The lower bound is attained if and only if exactly one minor is non-zero, and the
upper bound if and only if all minors have the same absolute value. Note that,
while the lower bound is the same for every architecture, the upper bound depends
on $N_k$, and hence on the number of hidden neurons $m$. We return to this point in
Section~\ref{sec: novel metric}.

The value of $\mathrm{PR}^{(k)}(x)$ can therefore be interpreted as follows:
\begin{itemize}
    \item \textbf{\boldmath Values close to $1$} indicate high concentration: a few minors
    carry most of the total magnitude, and $\mathrm{PR}^{(k)}(x) = 1$ when a single
    minor carries all of it. This is the kind of pattern expected from the KA
    construction.
    \item \textbf{\boldmath Values close to $\sqrt{N_k}$} indicate low concentration: the
    magnitude is spread almost uniformly across the minors, as expected in the
    absence of KA structure.
\end{itemize}

Finally, observe that if $J^{(k)}(x) = 0$, both norms vanish and Equation~\eqref{eq: PR} is not defined. This can happen even when the Jacobian itself is non-zero, for instance when the rank of $J_{\Phi^{\mathrm{MLP}}}(x)$ is smaller than $k$. For such points we adopt the convention $\mathrm{PR}^{(k)}(x) = \sqrt{N_k}$, which corresponds to the absence of concentration and is consistent with the convention that will be adopted for the $\alpha$-RC metrics in Definition~\ref{def: renyi_concentration_metric}.

To illustrate this metric in practice, the following example demonstrates how the Participation Ratio is computed for two contrasting scenarios.

\begin{example}[Calculation of the Participation Ratio] \label{ex: PR}
Consider two cases with $N_k = 4$ minors of order $k$ evaluated at a point $x$,
whose absolute values are
\begin{equation*}
|J_A^{(k)}(x)| = \begin{pmatrix} 1 & 1 \\ 1 & 1 \end{pmatrix}, \qquad
|J_B^{(k)}(x)| = \begin{pmatrix} 4 & 0 \\ 0 & 0 \end{pmatrix}.
\end{equation*}
Both matrices have the same $\ell_1$ norm,
$\|J_A^{(k)}(x)\|_1 = \|J_B^{(k)}(x)\|_1 = 4$, but their $\ell_2$ norms differ:
$\|J_A^{(k)}(x)\|_2 = 2$ and $\|J_B^{(k)}(x)\|_2 = 4$. Therefore
\begin{equation*}
\mathrm{PR}_A^{(k)}(x) = \frac{4}{2} = 2 = \sqrt{4}, \qquad
\mathrm{PR}_B^{(k)}(x) = \frac{4}{4} = 1.
\end{equation*}
In the first case the magnitude is spread uniformly across all four minors, and the
PR attains its maximum value $\sqrt{N_k} = 2$. In the second case all the magnitude
is concentrated in a single minor, and the PR attains its minimum \mbox{value $1$.}
\end{example}

\section{A novel improved metric}
\label{sec: novel metric}

In this section we propose a novel metric that circumvents the limitations of the Participation Ratio metric to assess the concentration of the matrix $J^{(k)}(x)$. 

This section is organized in two parts. Section~\ref{subs: motivating example} revisits the Participation Ratio and illustrates, through a simple example, why its dependence on the network architecture prevents comparing minor concentration across networks of different width. Section~\ref{subs: renyi metric} then introduces the proposed Rényi-based concentration metric: it first normalizes the exterior power values, recalls the Rényi divergence and its main properties, and formally defines the resulting class of $\alpha$-RC metrics. It then shows that the particular case $\alpha=2$ is directly related to the Participation Ratio, and closes with an illustrative example.

\subsection{Motivating example}\label{subs: motivating example}

As established in the previous section, the Participation Ratio $\mathrm{PR}^{(k)}(x)$ proposed by \citep{freedman2025spontaneous} possesses several appealing properties for quantifying higher-order interactions. It offers an intuitive geometric interpretation of active minors, has well-defined theoretical bounds ($1 \le \mathrm{PR}^{(k)}(x) \le \sqrt{N_k}$), and 
captures the degree of minor concentration at an individual input point $x$. 

Despite these strengths, the Participation Ratio suffers from a fundamental limitation: its range and its interpretation in the upper extreme limit are architecture-dependent. Because the theoretical upper bound $\sqrt{N_k} = \sqrt{\binom{m}{k}\binom{n}{k}}$ depends explicitly on the number of hidden units $m$, the scale of $\mathrm{PR}^{(k)}(x)$ varies dynamically with the network capacity. Consequently, the same numerical PR value can carry completely opposite geometric meanings depending on the size of the network.

To clearly illustrate this scale-dependency issue, consider the following example.

\begin{example}[Incompatibility of PR across architectural scales] \label{ex: motivating example}
Suppose we evaluate the order-$k$ minor concentration of two trained neural networks with different hidden layer widths at an input point $x$:

\begin{itemize}
    \item \textbf{Network A (Smaller architecture, $\boldsymbol{N_k = 2}$):} The matrix of order-$k$ minors is given by:
    \[
    J_A^{(k)}(x) = \begin{pmatrix} 1 & 1 \end{pmatrix},
    \]
    where $N_k = 2$ total minors exist. Computing the norms yields $\|J_A^{(k)}(x)\|_1 = 2$ and $\|J_A^{(k)}(x)\|_2 = \sqrt{1^2 + 1^2} = \sqrt{2}$, yielding
    \[
    \mathrm{PR}_A^{(k)}(x) = \frac{2}{\sqrt{2}} = \sqrt{2}.
    \]
    Here, activity is uniformly distributed across all possible minors ($\mathrm{PR}_A^{(k)}(x) = \sqrt{N_k}$), representing maximum dispersion (zero geometric concentration).

    \item \textbf{Network B (Larger architecture, $\boldsymbol{N_k = 16}$):} The matrix of order-$k$ minors is given by:
    \[
    J_B^{(k)}(x) = \begin{pmatrix} 
    1 & 1 & 0 & 0 \\ 
    0 & 0 & 0 & 0 \\ 
    0 & 0 & 0 & 0 \\ 
    0 & 0 & 0 & 0 
    \end{pmatrix},
    \]
    where $N_k = 16$ total minors exist. Computing the norms yields $\|J_B^{(k)}(x)\|_1 = 2$ and $\|J_B^{(k)}(x)\|_2 = \sqrt{1^2 + 1^2} = \sqrt{2}$, giving
    \[
    \mathrm{PR}_B^{(k)}(x) = \frac{2}{\sqrt{2}} = \sqrt{2}.
    \]
    Here, activity is concentrated in only 2 out of 16 minors, representing strong geometric concentration relative to the network's capacity (far below its maximum bound of $\sqrt{16} = 4$).
\end{itemize}

Despite operating under drastically different regimes, total uniform dispersion in Network A versus localized specialization in Network B, both examples yield the exact same Participation Ratio ($\mathrm{PR}^{(k)}(x) = \sqrt{2}$). Therefore, raw PR values fail to provide a meaningful basis for comparison across different scenarios.
\end{example}


The main motivation of this work is finding a concentration metric more flexible than PR with a simpler interpretation across architectures. To address both issues, we propose an alternative metric based on the \emph{Rényi entropy of order $\alpha$}. This metric is normalized so that its extreme values have the same meaning for every architecture, regardless of $N_k$. On the one hand, for every order $\alpha > 0$, the value $0$ corresponds to a uniform distribution of the magnitude across all minors, whereas the value $1$ is associated with a single non-zero minor. Concentration is thus measured relative to the maximum dispersion attainable by each network, rather than on a scale that changes with its width. It also has a parameter $\alpha$ which provides flexibility in the interpretation of such concentration. In the particular case of $\alpha=2$ this new metric is directly  expressible as a function of Freedman's original Participation Ratio.


\subsection{Rényi-based concentration metric}\label{subs: renyi metric}


In Section~\ref{subs: motivating example}, we analysed the PR as a measure of activity distribution across minors. However, as commented above, PR exhibits a fundamental quantitative limitation when evaluating minor concentration across different multilayer perceptron architectures.
In particular, the maximum value of $\mathrm{PR}^{(k)}(x)$ is given by $\sqrt{N_k}$, where $N_k = \binom{m}{k}\binom{n}{k}$ denotes the total number of minors of \mbox{order $k$.} Consequently, changing the width $m$ of a hidden layer alters the baseline scale of $\mathrm{PR}$, making direct comparisons between different network architectures impossible.

To overcome this limitation, we propose a novel concentration metric based on Rényi entropy of order $\alpha$ \citep{renyi1961measures}. This metric maps the activity distribution to a normalized concentration score bounded within the interval $[0,1]$ for any network architecture. Also, the value zero is interpreted as minimum concentration (uniform distribution) and the value one always as maximum concentration (only one non-zero element).

\subsubsection{Normalization of Exterior Power Values}

To apply an information-theoretic formulation, we first must interpret the matrix of minors $J^{(k)}(x)$ as a discrete probability distribution. This point of view is key for our goals, since the appropriate normalization of the absolute values of the entries of $J^{(k)}(x)$ can be viewed as a two-dimensional discrete \emph{Probability Mass Function} (PMF).

Given an input point $x \in \mathbb{R}^n$, let $N_k = \binom{m}{k}\binom{n}{k}$ be the total number of minors of order $k$, indexed by the multi-index sets $J \in \mathcal{J}_{k, m}$ and $I \in \mathcal{I}_{k, n}$. Normalizing the absolute values of the entries of $J^{(k)}(x)$ by its $\ell_1$-norm yields the relative weight $p_{J, I}^{(k)}(x)$ associated with the minor $J_{J, I}^{(k)}(x)$:
\begin{equation} \label{eq: probabilities}
    p_{J, I}^{(k)}(x) = \frac{\left|J_{J, I}^{(k)}(x)\right|}{\|J^{(k)}(x)\|_1} = \frac{\left|J_{J, I}^{(k)}(x)\right|}{\sum_{J' \in \mathcal{J}_{k, m}} \sum_{I' \in \mathcal{I}_{k, n}} \left|J_{J', I'}^{(k)}(x)\right|}.
\end{equation}
This collection can be interpreted as a discrete PMF $p^{(k)}(x) = \{p_{J, I}^{(k)}(x)\}_{J \in \mathcal{J}_{k, m}, I \in \mathcal{I}_{k, n}}$ over the $N_k$ entries of the exterior power. If all minors possess equal magnitude, $p^{(k)}(x)$ reduces to the uniform two-dimensional PMF with entries $1/N_k$. Conversely, if a single minor accumulates all magnitude, the two-dimensional PMF becomes degenerate.

\subsubsection{Rényi Divergence and Its Properties}

The Rényi divergence relies on the Rényi entropy. The Rényi entropy of order $\alpha\in[0,\infty)\cup\{\infty\}$ generalizes the Shannon entropy of a discrete probability distribution $p$ with $N$ entries. When  $\alpha \ge 0$ and $\alpha \neq 1$  the Rényi entropy is defined from the set of probabilities $\{p_i\}_{i = 1}^N$ as:
\begin{equation}
    H_\alpha(p) = \frac{1}{1-\alpha} \log \sum_{i=1}^N p_i^\alpha,
\end{equation}
where $p_i^0=1$ unless $p_i=0$ (in such a case $0^0=0$, by definition). Note that $H_0(p)=\log(|\textrm{supp}(p)|)$ is the Hartley entropy, with $\textrm{supp}(p)=\{i : p_i>0\}$ denoting the support of $p$.
Otherwise, when $\alpha\in\{1,\infty\}$ it is defined as the limit
\begin{equation}
    H_\alpha(p) = \lim_{\gamma\rightarrow\alpha}{H_\gamma(p)}.
\end{equation}
Specifically, $H_1(p)$ is the classical Shannon entropy, i.e. $H_1(p)=H_{\text{Sh}}(p) = -\sum_{i=1}^N p_i \log p_i$. The case $\alpha=\infty$ corresponds to the min-entropy $H_{\infty}(p)=-\log(\max_{1 \le i \le N} p_i)$ where the metric is determined by the information associated with the most likely outcome. Finally, note that the case $\alpha=2$ is widely used in practice and is known as the collision entropy. 

The parameter $\alpha$ regulates the weight assigned to outcomes of varying probabilities. The following cases are relevant for our discussion:
\begin{itemize}
    \item \textbf{Extreme choices ($\boldsymbol{\alpha = 0}$ and $\boldsymbol{\alpha = \infty}$):} For $\alpha = 0$, all non-zero entries are weighted equally regardless of magnitude. For $\alpha = \infty$, only the maximal element (that is to say, the mode of the discrete probability distribution) influences the value.
    \item \textbf{Choice of $\boldsymbol{\alpha = 1}$:} In this case we have the Shannon entropy as an intermediate choice.
    \item \textbf{Choice of $\boldsymbol{\alpha = 2}$:} The intermediate case $\alpha = 2$ yields the entropy:
    \begin{equation} \label{eq: renyi entropy of order 2}
        H_2(p) = - \log \sum_{i=1}^N p_i^2.
    \end{equation}
    The choice of $\alpha = 2$ provides a balanced weighting across the distribution while establishing a direct analytical bridge to the Participation Ratio, as we will describe later.
\end{itemize}

Let us define the Rényi divergence of order $\alpha$ (where $\alpha\in[0,\infty)\cup\{\infty\}$) of two discrete probability distributions $p$ and $q$ with $N$ outcomes and $\textrm{supp}(p)\subseteq\textrm{supp}(q)$. For $\alpha\in(0,1)\cup(1,\infty)$ the Rényi divergence is defined as \citep[Equation~(1)]{vanerven2014renyi}:
\begin{equation}
    \label{def_div_renyi}
    D_{\alpha}(p \parallel q) = \frac{1}{1-\alpha} \log \sum_{i=1}^N \frac{p_i^\alpha}{q_i^{\alpha-1}},
\end{equation}
and for the remaining cases, i.e. $\alpha\in\{0,1,\infty\}$, it is defined as the limit of \mbox{Equation~\eqref{def_div_renyi},}
\begin{equation}
    D_{\alpha}(p\parallel q)=\lim_{\gamma\rightarrow\alpha}{D_{\gamma}(p\parallel q)}.
\end{equation}
By taking limits using L'H\^{o}pital's Rule or evaluating directly, the expressions in Table~\ref{tab:renyi_casos_especiales} are obtained for these three special cases: the min-divergence (related to the Hartley entropy) for $\alpha=0$, the Kullback--Leibler divergence for $\alpha=1$, and the max-divergence for $\alpha=\infty$.

\begin{table}[H]
  \centering
  \begin{tabular}{c l l}
    \toprule
    Parameter ($\alpha$) & Expression & Name \\
    \midrule
    $\alpha = 0$ & $D_0(p \parallel q) = -\log \left( \sum_{i: p_i > 0} q_i \right)$ & Min-divergence/Hartley \\
    \addlinespace
    $\alpha = 1$ & $D_1(p \parallel q) = \sum_{i=1}^N p_i \log \left( \frac{p_i}{q_i} \right)$ & Kullback--Leibler ($D_{\text{KL}}$) \\
    \addlinespace
    $\alpha = \infty$ & $D_\infty(p \parallel q) = \log \left( \max_{1 \le i \le N} \frac{p_i}{q_i} \right)$ & Max-divergence \\
    \bottomrule
  \end{tabular}
\caption{Special cases of the Rényi divergence.}
\label{tab:renyi_casos_especiales}
\end{table}

The Rényi divergence is not a distance in the strict sense since, in general, it is neither symmetric nor satisfies the triangle inequality. However, it is nonnegative for every order $\alpha$, and for $\alpha > 0$ it also satisfies the identity of indiscernibles, that is, it vanishes if and only if both distributions coincide. For $\alpha = 0$ this property holds in a weaker form, since the divergence vanishes if and only if both distributions have the same support. These properties, established in Proposition~\ref{prop_Renyi_1}, are what we need for our objectives. 
\begin{proposition}
\label{prop_Renyi_1}
    For any $\alpha\in[0,\infty)\cup\{\infty\}$ and any discrete probability distributions $p$ and $q$ with $N$ outcomes and $\textrm{supp}(p)\subseteq\textrm{supp}(q)$, the Rényi divergence is nonnegative, i.e. $D_{\alpha}(p \parallel q) \ge 0$. Besides, for any $\alpha\in(0,\infty)\cup\{\infty\}$, $D_{\alpha}(p \parallel q) = 0$ if and only if $p = q$. The case $\alpha=0$ is special and the condition becomes $D_{0}(p \parallel q) = 0$ if and only if $\textrm{supp}(p)=\textrm{supp}(q)$.
\end{proposition}
\begin{proof} 
For $\alpha\in(0,1)\cup(1,\infty)\cup\{\infty\}$, the proof is found in \citep[Theorem~8]{vanerven2014renyi}. 

For $\alpha=1$, the Rényi divergence is the Kullback--Leibler divergence and the proof is found in \citep{cover2006elements}. 

From \mbox{Table~\ref{tab:renyi_casos_especiales},} the nonnegativity of the min-divergence ($\alpha=0$) is clear since {$\textrm{supp}(p)\subseteq\textrm{supp}(q)$} and consequently $\sum_{i: p_i > 0} q_i\in (0,1]$. In this case, $\alpha=0$, using Table~\ref{tab:renyi_casos_especiales} we have that, if $\textrm{supp}(p)=\textrm{supp}(q)$, then all the probability mass of $q$ lies on the same indices as that of $p$, and therefore $\sum_{i: p_i > 0} q_i=1$, so that $D_{0}(p \parallel q)=-\log(1)=0$. On the other hand, if $D_{0}(p \parallel q) = 0$, then $\sum_{i: p_i > 0} q_i = 1$, so that $q_i = 0$ for every $i \notin \textrm{supp}(p)$, that is, $\textrm{supp}(q)\subseteq\textrm{supp}(p)$, which together with the hypothesis $\textrm{supp}(p)\subseteq\textrm{supp}(q)$ gives $\textrm{supp}(p)=\textrm{supp}(q)$.
\end{proof}

The expression of the Rényi divergence takes a simplified form when the second distribution is the uniform distribution $u=(1/N,\dots,1/N)$, which can be written in terms of the Rényi entropy, as stated in the following proposition.
\begin{proposition}
\label{prop_Renyi_2}
    For any $\alpha\in[0,\infty)\cup\{\infty\}$ and any discrete probability distribution $p$ with $N$ outcomes, let us consider the uniform distribution $u = (1/N, \dots, 1/N)$, then
    \begin{equation}
        D_{\alpha}(p \parallel u) = \log N - H_{\alpha}(p).
    \end{equation}
\end{proposition}
\begin{proof} 
As in Proposition~\ref{prop_Renyi_1}, for $\alpha\in(0,1)\cup(1,\infty)$ the proof is found in \citep[Equation~(2)]{vanerven2014renyi}, while for $\alpha=1$ it is found in \citep{cover2006elements}. 

For the case $\alpha=0$, from Table~\ref{tab:renyi_casos_especiales} we obtain 
\begin{equation}
D_0(p \parallel u) = -\log \left( \sum_{i: p_i > 0} \frac{1}{N} \right)= -\log \left( |\textrm{supp}(p)| \frac{1}{N} \right)=\log N - H_0(p). 
\end{equation}
Finally, for the case $\alpha=\infty$ and again from Table~\ref{tab:renyi_casos_especiales} we obtain
\begin{equation}
D_{\infty}(p \parallel u) = \log \left( \max_{1 \le i \le N} \frac{p_i}{1/N} \right)= \log \left(N\max_{1 \le i \le N} p_i\right)=\log N - H_{\infty}(p). 
\end{equation}
\end{proof}

With the aid of Proposition~\ref{prop_Renyi_2}, the Rényi divergence of order \mbox{$\alpha\in[0,\infty)\cup\{\infty\}$} allows us to quantify the \emph{distance} (not a metric distance) between a certain distribution $p$ and the uniform $u$.
Specifically, for a fixed $\alpha\in[0,\infty)\cup\{\infty\}$ we have two extreme situations:
\begin{enumerate}
    \item \textbf{Minimum concentration:} When $p = u$ (completely uniform distribution) then $H_{\alpha}(p) = \log N$, yielding $D_{\alpha}(p \parallel u) = 0$ (a consequence of \mbox{Proposition~\ref{prop_Renyi_1}}). For $\alpha > 0$ this is the only case with $D_{\alpha}(p \parallel u) = 0$, whereas for $\alpha = 0$ the divergence vanishes for every $p$ with $\textrm{supp}(p) = \textrm{supp}(u)$, that is, for every $p$ without zero entries.
    \item \textbf{Maximum concentration:} When $p$ is degenerate ($p_i = 1$ for some $i$ and $0$ elsewhere) it is straightforward to check that $H_{\alpha}(p) = 0$, reaching the upper bound $D_{\alpha}(p \parallel u) = \log N$.
\end{enumerate}

In what follows we will always use the Rényi divergence comparing a probability distribution $p$ with the uniform distribution $q=u$. Therefore, we will simply denote this measure by $D_\alpha$, omitting $p$ and $u$, that is, $D_\alpha \doteq D_\alpha(p \parallel u)$, where the distribution $p$ will be clear from the context.

\subsubsection{Formal Definition of the Class of $\alpha$-RC Metrics}

Let us formally define the proposed class of \emph{Rényi-based Concentration} (RC) metrics: 
\begin{definition}[Class of $\alpha$-RC metrics]
\label{def: renyi_concentration_metric}
Let us fix $\alpha\in[0,\infty)\cup\{\infty\}$. Let $x \in \mathbb{R}^n$ be an input point, and let $J^{(k)}(x)$ denote the $k$-th exterior power matrix with $N_k = \binom{m}{k}\binom{n}{k}$ total entries. The pointwise $\alpha$-RC metric of order $k$ at $x$, denoted by $C_{\alpha}^{(k)}(x)$, is defined as follows:
\begin{equation} \label{eq: renyi concentration in x}
    C_{\alpha}^{(k)}(x) \doteq  \frac{D_{\alpha}^{(k)}(x)}{\log N_k}= 1 - \frac{H_{\alpha}^{(k)}(x)}{\log N_k},
\end{equation}
where $H_{\alpha}^{(k)}(x)$ represents the Rényi entropy associated with the probability distribution $p^{(k)}(x)$ in Equation~\eqref{eq: probabilities}
and $D_{\alpha}^{(k)}(x)$ the Rényi divergence of the same distribution with respect to the uniform.
Furthermore, for an evaluation dataset $\mathcal{D} = \{x^{(1)}, \dots, x^{(S)}\}$, the dataset-level $\alpha$-RC, $C_{\alpha}^{(k)}\left(\mathcal{D}\right)$, or simply $C_{\alpha}^{(k)}$, is defined as the sample mean across $\mathcal{D}$:
\begin{equation} \label{eq: C_renyi mean}
    C_{\alpha}^{(k)} \doteq \frac{1}{S} \sum_{s=1}^S C_{\alpha}^{(k)}(x^{(s)}).
\end{equation}
By construction, $0 \le C_{\alpha}^{(k)} \le 1$. For $\alpha > 0$, $C_{\alpha}^{(k)}(x) = 0$ if and only if the minor magnitudes are uniformly distributed, whereas $C_{0}^{(k)}(x) = 0$ if and only if no minor vanishes, that is, $|\textrm{supp}(p^{(k)}(x))| = N_k$ (Proposition~\ref{prop_Renyi_1}). For every $\alpha$, $C_{\alpha}^{(k)}(x) = 1$ if and only if all the magnitude is concentrated in a single minor (by convention, $C_{\alpha}^{(k)}(x) = 0$ if $J^{(k)}(x) = 0$).
\end{definition}

Note that, from this point on, the quantities $D_\alpha$ and $H_\alpha$ depend on the order $k$ and on the input point $x$. Indeed, both are built from the probability distribution $p^{(k)}(x)$ in Equation~\eqref{eq: probabilities}, which in turn is obtained from the exterior power $J^{(k)}(x)$ of the Jacobian. This is the reason why the Rényi divergence and the Rényi entropy appear in the definition of the $\alpha$-RC metrics with the notation $D_\alpha^{(k)}(x)$ and $H_\alpha^{(k)}(x)$, respectively.

Within this class of $\alpha$-RC metrics there are four special cases to be highlighted for a fixed $x \in \mathbb{R}^n$. The $0$-RC metric is based on the min-divergence (or Hartley) and only depends on the number of non-zero entries of $J^{(k)}(x)$. At the opposite extreme is the $\infty$-RC metric, which is related to the modal mass (probability at the mode). In an intermediate situation we have the $1$-RC metric that corresponds to the KL-divergence (and the Shannon entropy), and also the $2$-RC metric that is related to the collision entropy. The $2$-RC metric is directly connected to the PR as shown below. 

The class of $\alpha$-RC metrics has a very nice behaviour with respect to the parameter $\alpha$. Specifically, all these metrics are arranged in a monotonic way with respect to $\alpha$, while their range is restricted to the interval $[0,1]$. This result is collected in the following proposition.

\begin{proposition}[Monotonicity of the class of $\alpha$-RC metrics]
\label{prop: renyi_monotonic}
Let us fix the parameter $k$ with $1\leq k\leq\min(m,n)$ and $N_k\geq 2$, where $m$ is the number of hidden neurons of the MLP defined in Section~\ref{subs: MLP and KAG} and $n$ its input dimension, and let $x \in \mathbb{R}^n$ with \mbox{$J^{(k)}(x) \neq 0$.} Let $C_{\alpha}^{(k)}(x)$ denote the pointwise $\alpha$-RC metric of order $k$ at $x$, defined in Equation~\eqref{eq: renyi concentration in x} from the probability distribution $p^{(k)}(x)$ in Equation~\eqref{eq: probabilities}. Let $\alpha_1,\alpha_2\in[0,\infty)\cup\{\infty\}$ with $\alpha_1<\alpha_2$. Then $C_{\alpha_1}^{(k)}(x)\leq C_{\alpha_2}^{(k)}(x)$, with equality if and only if $p^{(k)}(x)$ is uniform on its support.
\end{proposition}
\begin{proof}
The inequality follows directly from the monotonicity of the Rényi divergence \citep[Theorem~3]{vanerven2014renyi} and Equation~\eqref{eq: renyi concentration in x}. 

Regarding the equality, if $p^{(k)}(x)$ is uniform on its support, then $H_{\alpha}^{(k)}(x)=\log|\textrm{supp}(p^{(k)}(x))|$ for every $\alpha$, so $C_{\alpha}^{(k)}(x)$ does not depend on $\alpha$ and, in particular, $C_{\alpha_1}^{(k)}(x)=C_{\alpha_2}^{(k)}(x)$.

Conversely, for two orders in $(0,1)\cup(1,\infty)$, the proof of \citep[Theorem~3]{vanerven2014renyi} shows that the equality holds if and only if $p^{(k)}(x)$ is uniform on its support. If $\alpha_1$ or $\alpha_2$ is $0$, $1$ or $\infty$, the equality and the monotonicity make $D_\alpha$ constant on $[\alpha_1,\alpha_2]$, which contains two such orders, and the same conclusion follows.
\end{proof}

\subsubsection{Connection of the $2$-RC Metric with the Participation Ratio}
A key analytical property of $C_2^{(k)}(x)$ is its exact closed-form relation with $\mathrm{PR}^{(k)}(x) = \frac{\|J^{(k)}(x)\|_1}{\|J^{(k)}(x)\|_2}$. Using Equation~\eqref{eq: probabilities}, we can express the sum of squared probabilities as follows:
\begin{equation} \label{eq: sum of squared probabilities}
    \sum_{J \in \mathcal{J}_{k, m}} \sum_{I \in \mathcal{I}_{k, n}} (p_{J, I}^{(k)}(x))^2 = \sum_{J, I} \frac{\left|J_{J, I}^{(k)}(x)\right|^2}{\|J^{(k)}(x)\|_1^2} = \frac{\|J^{(k)}(x)\|_2^2}{\|J^{(k)}(x)\|_1^2} = \frac{1}{\left(\mathrm{PR}^{(k)}(x)\right)^2}.
\end{equation}
Taking the logarithm in Equation~\eqref{eq: sum of squared probabilities} yields:
\begin{equation}
    H_2^{(k)}(x) = -\log \left( \frac{1}{\left(\mathrm{PR}^{(k)}(x)\right)^2} \right) = 2 \log\left(\mathrm{PR}^{(k)}(x)\right).
\end{equation}
Substituting $H_2^{(k)}(x)$ back into $C_2^{(k)}(x)$ from Equation~\eqref{eq: renyi concentration in x} leads to:
\begin{equation} \label{eq: C2 PR}
    C_2^{(k)}(x) = 1 - \frac{2 \log\left(\mathrm{PR}^{(k)}(x)\right)}{\log N_k}.
\end{equation}
This identity demonstrates that the $2$-RC metric $C_2^{(k)}$ logarithmically scales and normalizes $\mathrm{PR}$, removing its range dependency on architecture size while preserving its intuitive geometric interpretation.

\subsubsection{Illustrative Example with the $2$-RC Metric}
To demonstrate how $C_2^{(k)}$ improves the metric interpretation with different architectures, we go back to Example~\ref{ex: motivating example} and consider two networks evaluated at an input point $x$:
\begin{itemize}
    \item \textbf{Network A ($\boldsymbol{N_k = 2}$ available minors):} The exterior power $J_A^{(k)}(x) = \begin{pmatrix} 1 & 1 \end{pmatrix}$ has a Participation Ratio $\mathrm{PR}_A^{(k)}(x) = \frac{2}{\sqrt{2}} = \sqrt{2}$. 
    
    The $2$-RC metric is:
    \begin{equation}
        C_{2, A}^{(k)}(x) = 1 - \frac{2 \log(\sqrt{2})}{\log 2} = 1 - \frac{\log 2}{\log 2} = 0.
    \end{equation}
    Activity is evenly split across all $2$ available minors, resulting in zero concentration.

    \item \textbf{Network B ($\boldsymbol{N_k = 16}$ available minors):} The exterior power $J_B^{(k)}(x)$ has $2$ non-zero entries equal to $1$ and $14$ zero entries. Here $\mathrm{PR}_B^{(k)}(x) = \frac{2}{\sqrt{2}} = \sqrt{2}$ (identical to Network A). However, the Rényi concentration is:
    \begin{equation}
        C_{2, B}^{(k)}(x) = 1 - \frac{2 \log(\sqrt{2})}{\log 16} = 1 - \frac{\log 2}{4 \log 2} = 0.75.
    \end{equation}
\end{itemize}
While $\mathrm{PR}$ assigns both networks the same value ($\sqrt{2}$), $C_2^{(k)}$ assigns $0.75$ to Network B, accurately reflecting that its activity is concentrated within a small fraction ($2/16$) of its total capacity.


\section{Computational experiments} \label{sec: experiments}

This section illustrates the behaviour of the proposed concentration metric on single-hidden-layer MLPs trained on a continuous version of the XOR problem. We first describe the experimental setup in Section~\ref{subs: experimental setup} and the predictive performance of the trained models in Section~\ref{subs: predictive performance}. We then analyse the concentration of the minors at a fixed input point for every combination of width $m$ and order $k$ (Section~\ref{subs: pointwise concentration}) and, finally, its evolution along training for the largest configuration, $m=32$ and $k=3$ (Section~\ref{subs: concentration during training}). The main findings are discussed in Section~\ref{subs: discussion}.

\subsection{Experimental setup} \label{subs: experimental setup}

We consider the regression problem studied by \citep{freedman2025spontaneous}, in
which the target is the continuous XOR function
\begin{equation} \label{eq: xor}
f^{\mathrm{xor}}(x) = \prod_{i=1}^{n} \sin(\pi x_i), \qquad x \in [-1, 1]^n,
\end{equation}
with input dimension $n = 3$. On $[-1, 1]^3$ the sign of $f^{\mathrm{xor}}(x)$ is
the product of the signs of the coordinates of $x$, so that, as in its Boolean
counterpart, the output depends jointly on all the input variables. For each random
seed, we generate a training set
$\mathcal{D}_{\mathrm{train}} = \{(x^{(s)}, y^{(s)})\}_{s=1}^{S_{\mathrm{train}}}$
with $S_{\mathrm{train}} = 4096$ points and a test set
$\mathcal{D}_{\mathrm{test}} = \{(x^{(s)}, y^{(s)})\}_{s=1}^{S_{\mathrm{test}}}$
with $S_{\mathrm{test}} = 1024$ points, independently of each other. In both sets,
each input $x^{(s)} \in [-1, 1]^3$ is drawn independently from the uniform
distribution on $[-1, 1]^3$, and the corresponding target is computed without noise
as $y^{(s)} = f^{\mathrm{xor}}(x^{(s)})$. The superscript $(s)$ therefore indexes
the points of the set under consideration, which will be clear from the context.

Following the formulation of Section~\ref{subs: MLP and KAG}, we train
single-hidden-layer MLPs $f^{\mathrm{MLP}}(x) = B\,\boldsymbol{\sigma}(Ax + a) + b$
with two hidden widths, $m \in \{8, 32\}$. Here
$\boldsymbol{\sigma} : \mathbb{R}^m \to \mathbb{R}^m$ is the multivariate activation
function, which applies the univariate activation $\sigma : \mathbb{R} \to \mathbb{R}$
to each component of its argument. As univariate activation we use the exact GeLU,
$\sigma(t) = t\,F(t)$, detailed in Section~\ref{subs: MLP and KAG}. The weight matrices are initialized as random orthogonal matrices \citep{saxe2014exact}. Specifically, $A \in \mathcal{M}_{m \times n}(\mathbb{R})$ is a random matrix with orthonormal columns and $B \in \mathcal{M}_{1 \times m}(\mathbb{R})$ is a random vector of unit norm, both obtained from the QR decomposition of a matrix with independent standard Gaussian entries. The biases $a$ and $b$ are initialized to zero. The models are trained by minimizing the mean squared error with the Adam optimizer \citep{kingma2015adam}, with learning rate $0.01$, $(\beta_1,\beta_2) = (0.9, 0.999)$ and mini-batches of $250$ points reshuffled at every epoch, for a fixed budget of $500$ epochs and without early stopping or explicit regularization. This batch size coincides with the critical batch size reported by \citep{freedman2025spontaneous} for this problem. Each width is trained with the four seeds $\{0, 1, 2, 3\}$, which determine the dataset, the initialization and the order of the mini-batches. For a given seed, both widths are trained and evaluated on exactly the same training and test points, with the same mini-batch order. All the quantities described below are measured at initialization (epoch $0$) and at the end of each of the $500$ epochs, that is, on $501$ states of each network.

Predictive performance is evaluated with the coefficient of determination, $R^2$,
\citep[Equation~(III.5)]{freedman2025spontaneous}. Given a set
$\mathcal{D} = \{(x^{(s)}, y^{(s)})\}_{s=1}^{S}$ of $S$ points, it is defined as:

\begin{equation} \label{eq: R2}
R^2 = 1 - \frac{\sum_{s=1}^{S} \left( y^{(s)} - \hat{y}^{(s)} \right)^2}
{\sum_{s=1}^{S} \left( y^{(s)} - \bar{y} \right)^2},
\end{equation}

\noindent where $\hat{y}^{(s)} = f^{\mathrm{MLP}}(x^{(s)})$ is the prediction of the model for
the input $x^{(s)}$, and $\bar{y} = \frac{1}{S} \sum_{s=1}^{S} y^{(s)}$ is the
sample mean of the targets in $\mathcal{D}$. We denote by $R^2_{\mathrm{train}}$
and $R^2_{\mathrm{test}}$ the values obtained with
$\mathcal{D} = \mathcal{D}_{\mathrm{train}}$ ($S = S_{\mathrm{train}}$) and
$\mathcal{D} = \mathcal{D}_{\mathrm{test}}$ ($S = S_{\mathrm{test}}$),
respectively, so that each of them uses the mean of the targets of its own set. A
perfect predictor yields $R^2 = 1$, the constant predictor $\hat{y}^{(s)} = \bar{y}$
yields $R^2 = 0$, and negative values correspond to predictions worse than that
constant.

The concentration of the minors is evaluated with the class of $\alpha$-RC metrics $C^{(k)}_\alpha(x)$ of Definition~\ref{def: renyi_concentration_metric}, computed from the distribution $p^{(k)}(x)$ of Equation~\eqref{eq: probabilities} and averaged over a dataset as in Equation~\eqref{eq: C_renyi mean}. We recall that $C^{(k)}_\alpha(x)$ takes values in $[0, 1]$ and is non-decreasing in $\alpha$ (Proposition~\ref{prop: renyi_monotonic}), that the case $\alpha = 1$ corresponds to the Shannon entropy and the case $\alpha = 2$ to the collision entropy, and that the $2$-RC metric is related to the Participation Ratio at every point through Equation~\eqref{eq: C2 PR}. The results obtained for $\alpha = 2$ can therefore be read as a normalized, logarithmic version of the Participation Ratio.

We evaluate four orders, $\alpha \in \{1/2, 1, 2, \infty\}$. The $\infty$-RC metric is the upper limit of the class and only depends on the largest minor. The orders $\alpha = 1$ and $\alpha = 2$ correspond to the Shannon and collision entropies, respectively, the latter linked to the Participation Ratio through Equation~\eqref{eq: C2 PR}. The order $\alpha = 1/2$ gives more weight to the small minors than $\alpha = 1$. It is the only order for which the Rényi divergence is symmetric, and it is a function of the so-called Hellinger distance between $p^{(k)}(x)$ and the uniform distribution \citep{vanerven2014renyi}. Since $C^{(k)}_\alpha(x)$ is non-decreasing in $\alpha$ (Proposition~\ref{prop: renyi_monotonic}), the values for $\alpha = 1/2$ and $\alpha = \infty$ bound, at every point, those of any order between them.

The limit $\alpha = 0$ is not evaluated, because for the MLP of Section~\ref{subs: MLP and KAG} it carries no pointwise information. By Proposition~\ref{prop_Renyi_1} and Definition~\ref{def: renyi_concentration_metric}, $C^{(k)}_0(x)$ only depends on which minors are non-zero, and it vanishes whenever no minor is zero. By Equation~\eqref{eq: minors MLP}, a minor of the MLP vanishes if and only if $\det(A_{J,I}) = 0$ or $\sigma'(z_j(x)) = 0$ for some $j \in J$. The first condition does not depend on $x$. The derivative of GeLU only vanishes at the isolated point $t^\ast \approx -0.7518$, so the second condition only holds when $x$ lies on one of the hyperplanes $z_j(x) = t^\ast$, which form a set of zero measure. Hence, the support of $p^{(k)}(x)$ is the same at almost every $x$ and is determined by $A$, and so is $C^{(k)}_0(x)$. In particular, $C^{(k)}_0(x) = 0$ at almost every $x$ as soon as no $k \times k$ minor of $A$ vanishes, as is generically the case for random and trained weights. In our experiments, no computed minor is zero at any of the evaluated points, epochs and networks, so that $C^{(k)}_0(x) = 0$ in all of them. 

All the minors of orders $k \in \{1, 2, 3\}$ are computed at every point, without
sampling. For $n = 3$, the number of minors of each order, $N_k = \binom{m}{k}\binom{n}{k}$, is $N_1 = 24$, $N_2 = 84$ and $N_3 = 56$ for $m = 8$, and $N_1 = 96$, $N_2 = 1488$ and $N_3 = 4960$ for $m = 32$. As in Equation~\eqref{eq: C_renyi mean}, the concentration is first computed at each point
and then averaged, never computed from averaged minors. The computation exploits
the factorization of the minors of the MLP Jacobian given in
Equation~\eqref{eq: minors MLP},
$\big[J^{(k)}(x)\big]_{J,I} = \big(\prod_{j \in J} \sigma'(z_j(x))\big)
\det(A_{J,I})$, so that the determinants $\det(A_{J,I})$ are computed once per
network state and reused for all the points. All the computations are carried out in double precision. The derivatives $\sigma'(z_j(x))$ and the minors are handled in logarithmic scale, so that even the smallest minors are represented without underflow and with full relative precision.

Regarding the statistics reported below, $R^2$ is summarized by its mean $\pm$ standard deviation across the four seeds. For the concentration, the four networks of a given width are evaluated on their own test sets, which gives $4 \times 1024 = 4096$ test points in total. The concentration is summarized by the mean $\pm$ standard deviation of the pointwise values $C^{(k)}_\alpha(x)$ over these $4096$ test points, so that this dispersion combines the variability between points and between seeds. Note that these $4096$ test points, obtained by pooling the test sets of the four seeds ($4 \times 1024$), should not be confused with the $4096$ training points of each dataset ($S_{\mathrm{train}}=4096$). They are different sets of points, and the two numbers coincide only numerically. The implementation code, the synthetic datasets, the weights of the trained networks at every epoch and all the computed values of $C^{(k)}_\alpha(x)$ are publicly available on GitHub at \url{https://github.com/JuanCarlosFdezMlg/renyi-concentration-kag}.

\subsection{Predictive performance} \label{subs: predictive performance}

Table~\ref{tab: R2} reports the coefficient of determination given by Equation~\eqref{eq: R2} of both architectures on the training and test sets, at initialization and at the end of training.

\begin{table}[H]
    \centering
    \small
    \begin{tabular}{c l c c}
        \toprule
        \textbf{Epoch} & \textbf{Set} & $\boldsymbol{m = 8}$ & $\boldsymbol{m = 32}$ \\
        \midrule
        \multirow{2}{*}{0}
        & Training (4096) & $-0.2561 \pm 0.1741$ & $-0.0836 \pm 0.0611$ \\
         & Test (1024) & $-0.2358 \pm 0.1581$ & $-0.1056 \pm 0.0800$ \\
        \midrule
        \multirow{2}{*}{500}
        & Training (4096) & $0.6175 \pm 0.1686$ & $0.9322 \pm 0.0480$ \\
         & Test (1024) & $0.6011 \pm 0.1632$ & $0.9274 \pm 0.0516$ \\
        \bottomrule
    \end{tabular}
    \caption{Coefficient of determination $R^2$ (mean $\pm$ standard deviation across the four seeds) of the MLPs with $m = 8$ and $m = 32$ hidden neurons. The number of points of each set is given in parentheses.}
    \label{tab: R2}
\end{table}

At initialization, both architectures yield negative values of $R^2$, that is, their predictions are worse than the constant prediction $\bar{y}$. After $500$ epochs, the network with $m = 32$ reaches $R^2_{\mathrm{test}} = 0.9274 \pm 0.0516$, whereas the network with $m = 8$ remains at $R^2_{\mathrm{test}} = 0.6011 \pm 0.1632$. At the end of training, the values obtained on the training and test sets are almost identical, with differences below $0.017$ for both widths, so there is no sign of overfitting. Since both sets lead to the same conclusions, from now on we focus exclusively on the test sets, and all the quantities reported in the remainder of this section are computed on the test points only. The narrow network is also much more sensitive to the seed. Its test values range from $0.3869$ to $0.7592$, while those of the wide network range from $0.8519$ to $0.9660$. This is consistent with the larger number of parameters of the wider network, whose greater capacity tends to make the result of training less dependent on the initialization and on the order of the mini-batches.

\subsection{Pointwise concentration at a fixed input} \label{subs: pointwise concentration}

{\renewcommand{\baselinestretch}{1}\footnotesize
\renewcommand{\arraystretch}{1.1}
\setlength{\tabcolsep}{4pt}
\begin{longtable}{c c c c c c}
    \toprule
    & & \multicolumn{2}{c}{$m = 8$} & \multicolumn{2}{c}{$m = 32$} \\
    \cmidrule(lr){3-4} \cmidrule(lr){5-6}
    $k$ & Metric & Value at $x^\ast$ & Mean $\pm$ SD & Value at $x^\ast$ & Mean $\pm$ SD \\
\endfirsthead
    \toprule
    & & \multicolumn{2}{c}{$m = 8$} & \multicolumn{2}{c}{$m = 32$} \\
    \cmidrule(lr){3-4} \cmidrule(lr){5-6}
    $k$ & Metric & Value at $x^\ast$ & Mean $\pm$ SD & Value at $x^\ast$ & Mean $\pm$ SD \\
\endhead
    \bottomrule
    \multicolumn{6}{r}{\textit{Continued on next page}} \\
\endfoot
    \bottomrule
    \caption{$\alpha$-RC metrics $C^{(k)}_\alpha$ and Participation Ratio $\mathrm{PR}^{(k)}$ at the test point $x^\ast$ for the networks of seed $0$, and their mean $\pm$ standard deviation over the $4096$ test points of the four seeds.}\label{tab: C example} \\
\endlastfoot
    \midrule
    \multicolumn{6}{l}{\textit{Epoch 0 (initialization)}} \\*
    \addlinespace[2pt]
    \multirow{5}{*}{1} & $C^{(1)}_{1/2}$ & 0.0333 & $0.0695 \pm 0.0209$ & 0.0305 & $0.0392 \pm 0.0037$ \\*
    & $C^{(1)}_{1}$ & 0.0613 & $0.1202 \pm 0.0353$ & 0.0506 & $0.0687 \pm 0.0079$ \\*
    & $C^{(1)}_{2}$ & 0.1062 & $0.1908 \pm 0.0538$ & 0.0786 & $0.1122 \pm 0.0160$ \\*
    & $C^{(1)}_{\infty}$ & 0.3083 & $0.3998 \pm 0.0708$ & 0.2304 & $0.3037 \pm 0.0453$ \\*
    & $\mathrm{PR}^{(1)}$ & 4.1385 & $3.6306 \pm 0.3008$ & 8.1882 & $7.5901 \pm 0.2748$ \\
    \cmidrule(lr){1-6}
    \multirow{5}{*}{2} & $C^{(2)}_{1/2}$ & 0.0406 & $0.0742 \pm 0.0265$ & 0.0268 & $0.0373 \pm 0.0050$ \\*
    & $C^{(2)}_{1}$ & 0.0713 & $0.1294 \pm 0.0449$ & 0.0464 & $0.0667 \pm 0.0098$ \\*
    & $C^{(2)}_{2}$ & 0.1177 & $0.2058 \pm 0.0664$ & 0.0744 & $0.1120 \pm 0.0186$ \\*
    & $C^{(2)}_{\infty}$ & 0.3155 & $0.4179 \pm 0.0755$ & 0.2035 & $0.3049 \pm 0.0389$ \\*
    & $\mathrm{PR}^{(2)}$ & 7.0609 & $5.8711 \pm 0.8301$ & 29.3957 & $25.6831 \pm 1.7130$ \\
    \cmidrule(lr){1-6}
    \multirow{5}{*}{3} & $C^{(3)}_{1/2}$ & 0.0504 & $0.0934 \pm 0.0411$ & 0.0272 & $0.0393 \pm 0.0061$ \\*
    & $C^{(3)}_{1}$ & 0.0946 & $0.1639 \pm 0.0689$ & 0.0476 & $0.0714 \pm 0.0121$ \\*
    & $C^{(3)}_{2}$ & 0.1642 & $0.2559 \pm 0.0994$ & 0.0770 & $0.1216 \pm 0.0226$ \\*
    & $C^{(3)}_{\infty}$ & 0.3937 & $0.4605 \pm 0.1076$ & 0.2199 & $0.3201 \pm 0.0392$ \\*
    & $\mathrm{PR}^{(3)}$ & 5.3777 & $4.5578 \pm 0.8580$ & 50.7623 & $42.1749 \pm 3.9628$ \\
    \midrule\midrule
    \multicolumn{6}{l}{\textit{Epoch 500 (end of training)}} \\*
    \addlinespace[2pt]
    \multirow{5}{*}{1} & $C^{(1)}_{1/2}$ & 0.0961 & $0.1401 \pm 0.0585$ & 0.2042 & $0.2478 \pm 0.0373$ \\*
    & $C^{(1)}_{1}$ & 0.1674 & $0.2195 \pm 0.0927$ & 0.2687 & $0.3163 \pm 0.0426$ \\*
    & $C^{(1)}_{2}$ & 0.2499 & $0.3018 \pm 0.1230$ & 0.3303 & $0.3752 \pm 0.0444$ \\*
    & $C^{(1)}_{\infty}$ & 0.4051 & $0.4525 \pm 0.1331$ & 0.4433 & $0.4838 \pm 0.0442$ \\*
    & $\mathrm{PR}^{(1)}$ & 3.2937 & $3.0888 \pm 0.5711$ & 4.6109 & $4.1834 \pm 0.4262$ \\
    \cmidrule(lr){1-6}
    \multirow{5}{*}{2} & $C^{(2)}_{1/2}$ & 0.1139 & $0.2062 \pm 0.0579$ & 0.2647 & $0.3318 \pm 0.0534$ \\*
    & $C^{(2)}_{1}$ & 0.1811 & $0.3105 \pm 0.0813$ & 0.3523 & $0.4352 \pm 0.0644$ \\*
    & $C^{(2)}_{2}$ & 0.2587 & $0.4084 \pm 0.0950$ & 0.4435 & $0.5204 \pm 0.0624$ \\*
    & $C^{(2)}_{\infty}$ & 0.4439 & $0.5641 \pm 0.0827$ & 0.5891 & $0.6412 \pm 0.0528$ \\*
    & $\mathrm{PR}^{(2)}$ & 5.1674 & $3.7889 \pm 0.7717$ & 7.6340 & $5.9160 \pm 1.3604$ \\
    \cmidrule(lr){1-6}
    \multirow{5}{*}{3} & $C^{(3)}_{1/2}$ & 0.1346 & $0.2619 \pm 0.0803$ & 0.3323 & $0.4198 \pm 0.0708$ \\*
    & $C^{(3)}_{1}$ & 0.2066 & $0.3695 \pm 0.0985$ & 0.4323 & $0.5426 \pm 0.0857$ \\*
    & $C^{(3)}_{2}$ & 0.2890 & $0.4673 \pm 0.1080$ & 0.5281 & $0.6384 \pm 0.0814$ \\*
    & $C^{(3)}_{\infty}$ & 0.4653 & $0.6313 \pm 0.0950$ & 0.6770 & $0.7556 \pm 0.0609$ \\*
    & $\mathrm{PR}^{(3)}$ & 4.1828 & $2.9900 \pm 0.6370$ & 7.4476 & $4.9413 \pm 1.6972$ \\
\end{longtable}
}

We first examine the concentration of the minors at an individual input point. We take as $x^\ast$ the first test point of seed $0$, $x^\ast = (0.1034, 0.1249, -0.0206)$, with target $f^{\mathrm{xor}}(x^\ast) = -0.0079$. Since the datasets are shared across widths, $x^\ast$ is the same point for $m = 8$ and $m = 32$. Table~\ref{tab: C example} reports $C^{(k)}_\alpha(x^\ast)$ for the networks of seed $0$, for all the widths, orders $k$ and values of $\alpha$, at initialization and at the end of training, together with the Participation Ratio, $\mathrm{PR}^{(k)}(x^\ast)$, of Section~\ref{subs: participation ratio}. Next to each value, the table also reports the mean $\pm$ standard deviation of the corresponding metric over the $4096$ test points of the four seeds, which places the value at $x^\ast$ within the distribution of the test set. Since the minors are computed exhaustively, $\mathrm{PR}^{(k)}$ is obtained at each point from $C^{(k)}_2$ through Equation~\eqref{eq: C2 PR}.

Several features of Table~\ref{tab: C example} are worth highlighting. First, the monotonicity of $C^{(k)}_\alpha$ with respect to $\alpha$ is observed in every block, both at $x^\ast$ and on average. For every width, order and epoch, $C^{(k)}_{1/2} \le C^{(k)}_1 \le C^{(k)}_2 \le C^{(k)}_{\infty}$. The order $\alpha$ therefore does not modify the information contained in the minors, but only the part of their distribution that is being summarized.

Second, training increases the concentration of every network, for every order $k$ and every value of $\alpha$, both at $x^\ast$ and on average, and the Participation Ratio decreases accordingly in every case. Within each network, the increase grows with the order of the minors. For $\alpha = 2$, the mean concentration of the network with $m = 32$ grows from $0.1122$ to $0.3752$ for $k = 1$, from $0.1120$ to $0.5204$ for $k = 2$ and from $0.1216$ to $0.6384$ for $k = 3$. For the network with $m = 8$, it grows from $0.1908$ to $0.3018$, from $0.2058$ to $0.4084$ and from $0.2559$ to $0.4673$, respectively.

Third, the rows of $\mathrm{PR}^{(k)}$ contain the same pointwise information as those of $C^{(k)}_2$, since both are related by the monotone transformation of Equation~\eqref{eq: C2 PR}. The difference lies in how each value is read. A value of the Participation Ratio can only be interpreted together with its upper bound $\sqrt{N_k}$, which depends on $m$ and $k$. This bound is $4.90$, $9.17$ and $7.48$ for $m = 8$, and $9.80$, $38.57$ and $70.43$ for $m = 32$, for $k = 1, 2, 3$, respectively. In contrast, $C^{(k)}_2$ directly expresses the position of the same value between the two extremes of its own range. For instance, for $m = 32$ and $k = 3$, $\mathrm{PR}^{(3)}(x^\ast)$ decreases from $50.7623$ to $7.4476$ out of a maximum of $70.43$, which corresponds to an increase of $C^{(3)}_2(x^\ast)$ from $0.0770$ to $0.5281$. Similarly, for $m = 8$ and $k = 3$, $\mathrm{PR}^{(3)}(x^\ast)$ decreases from $5.3777$ to $4.1828$ out of a maximum of $7.48$, and $C^{(3)}_2(x^\ast)$ increases from $0.1642$ to $0.2890$.

Finally, in every row of the table the point $x^\ast$ is less concentrated than the average test point, with a lower value of $C^{(k)}_\alpha$ and a higher value of $\mathrm{PR}^{(k)}$, both at initialization and at the end of training. The value at $x^\ast$ is therefore not representative of the whole test set by itself, which illustrates the pointwise heterogeneity summarized by the dataset-level metric of Equation~\eqref{eq: C_renyi mean} and by its standard deviation.

\subsection{Evolution of the concentration during training} \label{subs: concentration during training}

Although the previous analysis could be extended to every width and every order along all the epochs, from now on we restrict the study to $m = 32$ and $k = 3$ for the sake of brevity. This is the configuration with the best predictive performance and the largest number of minors, $N_3 = 4960$. For every epoch, the mean and the variance of $C^{(3)}_\alpha$ are computed over the $4096$ test points of the four seeds. Figure~\ref{fig: C epochs} shows the evolution of the mean concentration for the four values of $\alpha$, with a common vertical scale.

\begin{figure}[H]
    \centering
    \includegraphics[width=0.95\linewidth]{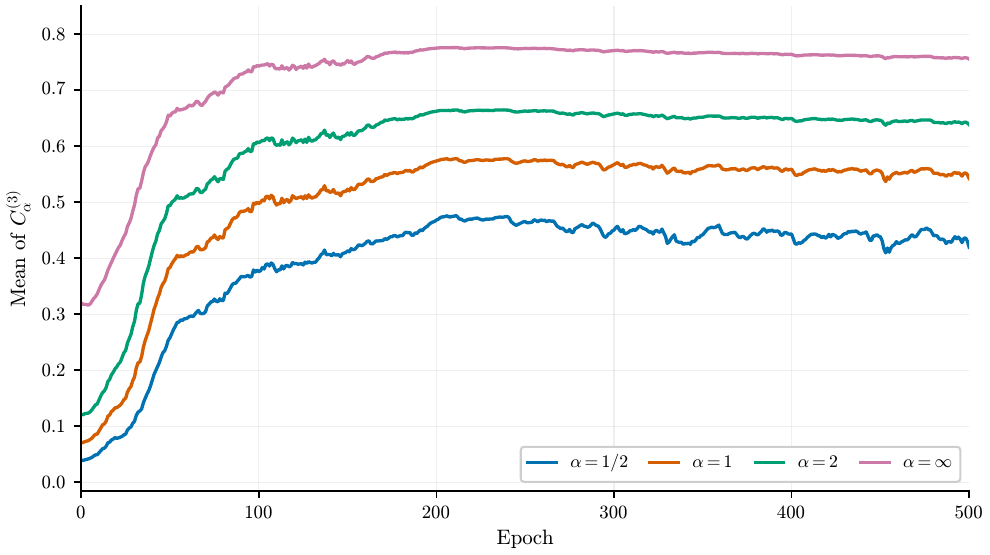}
    \caption{Mean of $C^{(3)}_\alpha$ over the test points during training for $m = 32$.}
    \label{fig: C epochs}
\end{figure}

At initialization, the mean values are $0.0393$, $0.0714$, $0.1216$ and $0.3201$ for $\alpha = 1/2, 1, 2$ and $\infty$, respectively. The concentration increases rapidly during the first $60$ epochs. For instance, the mean of $C^{(3)}_2$ exceeds $0.3$ at epoch $31$ and $0.5$ at epoch $52$. The growth then slows down until a maximum is reached between epochs $210$ and $240$, with values $0.4763$ for $\alpha = 1/2$ and $0.5778$ for $\alpha = 1$ (both at epoch $211$), $0.6649$ for $\alpha = 2$ (epoch $238$) and $0.7763$ for $\alpha = \infty$ (epoch $212$). From then on, the concentration decreases slightly, and at epoch $500$ the means are $0.4198$, $0.5426$, $0.6384$ and $0.7556$, respectively. As expected from the monotonicity in $\alpha$, the curves never cross, and the gap between $\alpha = 1$ and $\alpha = \infty$ remains of the same magnitude throughout training, from $0.25$ at initialization to $0.21$ at the end.

The maximum of the concentration roughly coincides with that of the number of hidden neurons whose input weights are very small, below $10^{-8}$ times the largest entry of $A$. On average over the four seeds, there is no such neuron before epoch $65$, their number reaches $13.25$ out of $32$ at epoch $227$, and then decreases to $5.25$ at epoch $500$. By Equation~\eqref{eq: minors MLP}, the minors that involve these neurons become negligible compared with the others, so that the magnitude concentrates on the minors of the remaining neurons.

The slight decrease of the concentration after its maximum agrees with the observations of \citep{freedman2025spontaneous}, who report an extremum of their minor concentration metrics during training near a training $R^2$ of about $0.85$, away from the point of peak performance (Figure~9 of that work). In our experiments, the maxima of $C^{(3)}_\alpha$ are reached at values of $R^2_{\mathrm{test}}$ between $0.84$ and $0.86$. As those authors point out, the cause of this extremum is not yet understood. We have checked that it is not caused by hidden neurons becoming active again. The fraction of hidden neurons that are locally inactive at the test points, with $|\sigma'(z_j(x))| < 0.05$, grows from $0.1\%$ at initialization to $11.4\%$ at epoch $500$, where it reaches its maximum.

Figure~\ref{fig: R2 magma} relates this evolution to the predictive performance of the same networks. Each panel corresponds to a value of $\alpha$ and contains one point per epoch. The vertical axis is the mean $R^2_{\mathrm{test}}$ of the four seeds, and the colour of each point is the mean of $C^{(3)}_\alpha$ at that epoch, on a magma scale fixed to $[0, 1]$ and common to the four panels. The halo around each point has an area proportional to the variance of $C^{(3)}_\alpha$ at that epoch, relative to the maximum variance reached in the panel. The legend of each panel shows the halos that correspond to one quarter, one half and the whole of that maximum variance. For instance, the maximum variance for $\alpha = 2$ is $0.0075$, so its legend shows $0.0019$, $0.0037$ and $0.0075$. The halo therefore encodes the dispersion of the concentration, not an interval for $R^2$.

\begin{figure}[H]
    \centering
    \includegraphics[width=\linewidth]{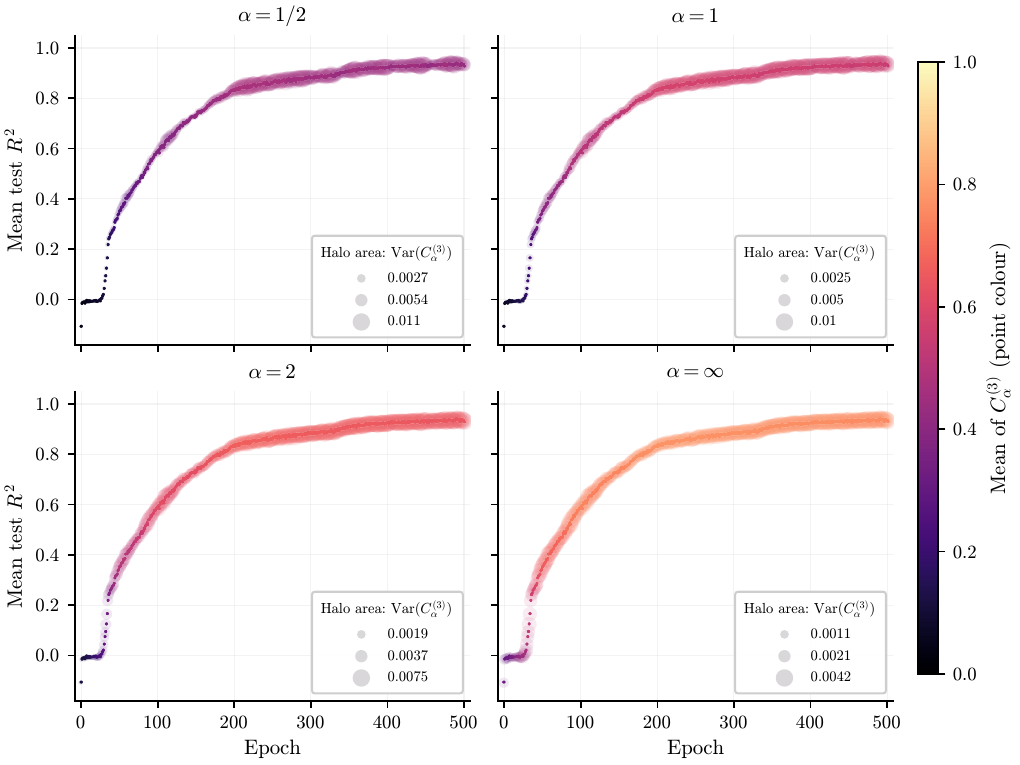}
    \caption{Joint evolution of the predictive performance and the concentration of the minors for $m = 32$ and $k = 3$ on the test sets.}
    \label{fig: R2 magma}
\end{figure}

During the first epochs, the mean $R^2_{\mathrm{test}}$ remains close to zero, with $R^2_{\mathrm{test}} = -0.0061$ at epoch $20$ and $0.0456$ at epoch $30$, and then rises sharply. It exceeds $0.25$ at epoch $38$, $0.5$ at epoch $82$, $0.75$ at epoch $158$ and $0.9$ at epoch $338$, reaching $0.9274$ at epoch $500$. The change of colour of the points, from dark to bright tones, takes place along this rise. It is worth noting that the concentration starts to grow before $R^2_{\mathrm{test}}$ departs from zero. The mean of $C^{(3)}_2$ increases from $0.1216$ at initialization to more than $0.2$ at epoch $19$, while the predictions are still no better than the constant $\bar{y}$. In the final phase, in contrast, $R^2_{\mathrm{test}}$ keeps increasing, from $0.8282$ at epoch $200$ to $0.9274$ at epoch $500$, while the concentration decreases slightly. As already observed by \citep{freedman2025spontaneous}, the relationship between both quantities is therefore not monotone along the whole training.

The halos show that the dispersion of the concentration also grows during training. For $\alpha = 2$, the standard deviation increases from $0.0226$ at initialization to $0.0814$ at epoch $500$, and the maximum variance for $\alpha = 1/2$, $\alpha = 1$ and $\alpha = 2$ is reached at epoch $486$. For $\alpha = \infty$, the maximum variance is reached earlier, at epoch $111$.

\subsection{Discussion} \label{subs: discussion}

The experiments illustrate how the proposed metric is read in practice. Each value of $C^{(k)}_\alpha$ lies in $[0, 1]$ and expresses the position of the distribution of minors between the two extremes attainable for its own number of minors, namely the uniform distribution and a single dominant minor. For the network with $m = 32$ and $k = 3$, for instance, this shows that the mean of $C^{(3)}_2$ evolves from $0.1216$, close to the uniform distribution, to $0.6384$, an intermediate level of concentration. The Participation Ratio carries the same pointwise information for $\alpha = 2$, but its values can only be interpreted together with the bound $\sqrt{N_k}$ of each configuration.

Regarding the relationship with predictive performance in the network with $m = 32$, the growth of the concentration accompanies, and slightly precedes, the onset of learning. However, the concentration is not a monotone function of the performance, since it decreases slightly in the last phase of training while $R^2_{\mathrm{test}}$ keeps improving. This agrees with the observation of \citep{freedman2025spontaneous} that the link between Kolmogorov--Arnold geometry and performance depends on the training regime.

The effect of the order $k$ can be interpreted through the factorization of Equation~\eqref{eq: minors MLP}. Each minor of order $k$ contains the product of the derivatives $\sigma'(z_j(x))$ of $k$ hidden neurons, so that when some of these derivatives become small at a point, the minors that involve the corresponding neurons are suppressed multiplicatively. This effect is amplified as $k$ grows, which is consistent with the larger increase in concentration observed for higher orders within each network. Even so, the concentration remains far from its maximum value. For $m = 32$ and $k = 3$, the mean of the $\infty$-RC metric $C^{(3)}_\infty$ never exceeds $0.7763$ and is $0.7556$ at epoch $500$, so even the most mode-sensitive member of the class stays well below $1$ and the magnitude is never carried by a single minor. At the other end of the evaluated orders, the mean of $C^{(3)}_{1/2}$, which gives more weight to the small minors, never exceeds $0.4763$ and is $0.4198$ at epoch $500$. The magnitude therefore concentrates on a subset of the minors without collapsing onto a single one, and no minor vanishes. This intermediate level of concentration is consistent with the observation of \citep{freedman2025spontaneous} that the proportion of zero rows in trained MLPs is well below that of the ideal Kolmogorov--Arnold construction.

Concerning the choice of $\alpha$, the qualitative conclusions above, namely the increase of the concentration with training and its larger increase for higher orders within each network, hold for the four evaluated orders. The order $\alpha$ determines which part of the distribution of minors is summarized. The order $\alpha = 1/2$ gives more weight to the small minors, the Shannon ($\alpha = 1$) and collision ($\alpha = 2$) entropies provide intermediate summaries, the latter directly linked to the Participation Ratio, and the limit $\alpha = \infty$ only depends on the largest minor. By monotonicity, the orders $\alpha = 1/2$ and $\alpha = \infty$ delimit, at every point, the values of any order between them. The limit $\alpha = 0$, which only depends on which minors vanish, is constant for the MLP with GeLU, as discussed in Section~\ref{subs: experimental setup}.


\section{Conclusions and Future Research}

This paper has introduced the class of $\alpha$-RC metrics to measure the concentration of the minors of the Jacobian of the inner map, a central notion of Kolmogorov--Arnold Geometry. The metrics are built from the R\'enyi divergence of order $\alpha$ between the normalized distribution of the minors and the uniform distribution, divided by its maximum value $\log N_k$. As a result, every member of the class takes values in $[0,1]$, and both extremes have the same meaning for every architecture. The value $0$ corresponds, for every $\alpha > 0$, to a uniform distribution of the minors, and the value $1$ to a single minor carrying all the magnitude. The class is non-decreasing in $\alpha$, with equality between two orders only when the minors are uniform on their support, and the order $\alpha$ selects which part of the distribution is summarized, from the small minors for $\alpha = 1/2$ to the largest minor for $\alpha = \infty$. For $\alpha = 2$ the metric is a closed-form transformation of the Participation Ratio, so it keeps the pointwise information of the Participation Ratio while removing its dependence on the number of minors.

The experiments with single-hidden-layer MLPs trained on a continuous XOR problem illustrate how the class is read in practice. Training increases the concentration for every width, every order $k$ and every value of $\alpha$, and the increase grows with $k$. For the network with $m = 32$ hidden neurons and $k = 3$, the mean of $C^{(3)}_2$ grows from $0.12$ at initialization to a maximum of $0.66$, and it starts to increase before the coefficient of determination departs from zero. The relationship with the predictive performance is not monotone, since the concentration decreases slightly in the last phase of training while the performance keeps improving, in agreement with the observations of \citep{freedman2025spontaneous}. The concentration reaches intermediate values, with the mean of $C^{(3)}_\infty$ below $0.78$ throughout training, so the magnitude concentrates on a subset of the minors without collapsing onto a single one.

The proposed family of metrics has been applied to networks in which all the minors can be enumerated. Several directions remain open. The number of minors $N_k = \binom{m}{k}\binom{n}{k}$ grows combinatorially with the width, the input dimension and the order $k$, so the exact computation of the metrics is not feasible for realistic networks, and efficient estimation methods are needed. The early growth of the concentration also suggests that the metrics could be used not only to monitor training but also to guide it, for instance by adding a concentration term to the loss function or by interleaving gradient steps on the concentration, as proposed by \citep{freedman2024proof}. 

\section*{CRediT authorship contribution statement}
\textbf{Juan Carlos Fernández:} Conceptualization, Methodology, Software, Writing – original draft, Writing – review \& editing. \textbf{Asunción Jiménez-Cordero:}  Conceptualization, Methodology, Supervision, Validation, Writing – original draft, Writing – review \& editing. \textbf{José Francisco Paris:} Conceptualization, Methodology,  Supervision, Validation, Writing – review \& editing. 

\section*{Declaration of Competing Interest}
The authors declare that they have no known competing financial interests or personal relationships that could have appeared to influence the work reported in this paper.

 \section*{Acknowledgements}
 The authors thankfully acknowledge the computer resources (Picasso Supercomputer), technical expertise and assistance provided by the SCBI (Supercomputing and Bioinformatics) center of the University of Málaga. J. F. Paris is partially supported by the project PAR 11/2023 of the University of Málaga.

\section*{Data availability}
The implementation code, the synthetic datasets, the weights of the trained networks at every epoch and all the computed values of $C^{(k)}_\alpha(x)$ are publicly available on GitHub at \url{https://github.com/JuanCarlosFdezMlg/renyi-concentration-kag}.

\section*{Declaration of generative AI and AI assisted technologies in the manuscript preparation process}
During the preparation of this work the authors used Claude (Anthropic) and Codex (OpenAI) in order to assist with the implementation of the code for the computational experiments, to draft and revise parts of the text, and to improve the language and readability of the manuscript. After using these tools, the authors reviewed and edited the content as needed and take full responsibility for the content of the publication.

\bibliography{sample}

\end{document}